\documentclass[creativecommons,sharealike]{eptcs}
\providecommand{\event}{PLACES 2020} 
\usepackage{underscore} 
\usepackage{easy-todo}
\usepackage[english]{babel}
\usepackage[table]{xcolor}
\usepackage{titlesec}
\usepackage{tabularx}
\usepackage{microtype}

\usepackage[T1]{fontenc}
\usepackage[utf8]{inputenc}
\usepackage[english]{babel}
\usepackage{csquotes}
\usepackage{amssymb}
\usepackage{amsmath}
\usepackage{relsize}
\usepackage{graphicx}
\usepackage{bm}
\usepackage{bussproofs}
\usepackage{mathtools}
\usepackage{amsthm}
\usepackage{stmaryrd}
\usepackage{xspace}

\usepackage[inline]{enumitem}

\newcommand{\ill}{\ensuremath{\mathsf{ILL}}\xspace}
\newcommand{\pill}{\ensuremath{\pi\mathsf{ILL}}\xspace}
\newcommand{\cll}{\ensuremath{\mathsf{CLL}}\xspace}
\newcommand{\pcll}{\ensuremath{\pi\mathsf{CLL}}\xspace}
\newcommand{\cp}{\ensuremath{\mathsf{CP}}\xspace}
\newcommand{\pcllcp}{\ensuremath{\pcll^\cp}\xspace}
\newcommand{\ull}{\ensuremath{\mathsf{ULL}}\xspace}
\newcommand{\pull}{\ensuremath{\pi\mathsf{ULL}}\xspace}
\newcommand{\vdill}{\ensuremath{\vdash_\ill}}
\newcommand{\vdcll}{\ensuremath{\vdash_\cll}}
\newcommand{\U}{\ensuremath{\mathfrak{U}}}
\newcommand{\C}{\ensuremath{\mathfrak{C}}}
\newcommand{\I}{\ensuremath{\mathfrak{I}}}

\newcommand{\infrule}[1]{\ensuremath{\text{\textsc{#1}}}}
\newcommand{\lft}{\ensuremath{\text{\textsc{L}}}}
\newcommand{\rgt}{\ensuremath{\text{\textsc{R}}}}
\newcommand{\intrule}{\ensuremath{\prescript{\ast}{}}}

\newcommand{\id}{\ensuremath{\infrule{Id}}}
\newcommand{\cut}{\ensuremath{\infrule{Cut}}}
\newcommand{\cpy}{\ensuremath{\infrule{Copy}}}

\newcommand{\infAx}[2]{\AxiomC{} \RightLabel{(#2)} \UnaryInfC{#1}}
\newcommand{\infAss}[1]{\AxiomC{#1}}

\newcommand{\infUn}[2]{\RightLabel{(#2)} \UnaryInfC{#1}}
\newcommand{\infBin}[2]{\RightLabel{(#2)} \BinaryInfC{#1}}

\let\onu\nu
\renewcommand{\nu}[1]{\ensuremath{\textstyle \mathsmaller{(\onu #1)}}}
\let\oparallel\|
\renewcommand{\|}{\ensuremath{\mathbin{|}}}
\newcommand{\fwd}[2]{\ensuremath{[#1 \mathbin{\leftrightarrow} #2]}}
\newcommand{\0}{\ensuremath{\bm{0}}}
\newcommand{\send}[2]{\ensuremath{#1 \langle #2 \rangle}}
\newcommand{\recv}[2]{\ensuremath{#1(#2)}}
\newcommand{\serv}[2]{\ensuremath{!#1(#2)}}
\newcommand{\subst}[1]{\ensuremath{\{#1\}}}

\newcommand{\tensor}{\ensuremath{\mathbin{\otimes}}\xspace}
\newcommand{\lolli}{\ensuremath{\mathbin{\multimap}}\xspace}
\newcommand{\parr}{\ensuremath{\mathbin{\rotatebox[origin=c]{180}{\&}}}\xspace}

\newcommand{\1}{\ensuremath{\bm{1}}}
\newcommand{\sbot}{\ensuremath{\mathsmaller{\bot}}}
\newcommand{\bang}{\ensuremath{{!}}\xspace}
\newcommand{\whynot}{\ensuremath{{?}}\xspace}
\newcommand{\dual}[1]{\ensuremath{#1^\sbot}}

\newcommand{\fn}{\ensuremath{\text{fn}}}
\newcommand{\bn}{\ensuremath{\text{bn}}}
\newcommand{\red}{\ensuremath{\mathbin{\rightarrow}}}
\newcommand{\bnf}{\ensuremath{\textstyle \mathbin{\mathlarger{\mathlarger{|}}}}}

\newcommand{\rdegree}{$r$-degree\xspace}
\newcommand{\live}{\ensuremath{\text{live}}}

\newcommand{\tend}{\ensuremath{\mathsf{end}}}

\newtheorem{theorem}{Theorem}[section]
\newtheorem{lemma}[theorem]{Lemma}

\newtheorem{definition}{Definition}[section]

\usepackage{hyperref}
\hypersetup{
    hypertexnames=true,
}
\usepackage{cite}

\titlespacing*{\paragraph}{0pt}{1.25ex plus 1ex minus .2ex}{1em}

\newcommand{\thx}{\thanks{
    Work partially supported by the Netherlands Organization for Scientific
    Research (NWO) under the VIDI Project No. 016.Vidi.189.046 (Unifying
    Correctness for Communicating Software).
}}

\title{Session Type Systems based on Linear Logic: \\ Classical versus
Intuitionistic\thx}
\author{Bas van den Heuvel
    \qquad\qquad Jorge A. Pérez
\institute{University of Groningen, The Netherlands}
}
\def\titlerunning{Session Type Systems based on Linear Logic: Classical versus
Intuitionistic}
\def\authorrunning{B.\,van den Heuvel, J.A.\,Pérez}

\begin{document}

\maketitle

\begin{abstract}
    Session type systems have been given logical foundations via Curry-Howard
    correspondences based on both intuitionistic and classical linear logic.
    The type systems derived from the two logics enforce communication
    correctness on the same class of $\pi$-calculus processes, but they are
    significantly different. Caires, Pfenning, and Toninho informally observed
    that, unlike the classical type system, the intuitionistic type system
    enforces locality for shared channels, i.e.\ received channels cannot be
    used for replicated input. In this paper, we revisit this observation from a
    formal standpoint. We develop United Linear Logic (\ull), a logic
    encompassing both classical and intuitionistic linear logic. Then, following
    the Curry-Howard correspondences for session types, we define \pull, a
    session type system for the $\pi$-calculus based on \ull. Using \pull we can
    formally assess the difference between the intuitionistic and classical type
    systems, and justify the role of locality and symmetry therein.
\end{abstract}

\section{Introduction}

\emph{Session types} are a popular approach to typing message-passing
concurrency~\cite{honda_types_1993,takeuchi_interaction-based_1994,%
honda_language_1998}. They describe communication over channels as sequences of
communication actions. This way, e.g., the session type
$!\mathsf{int}.?\mathsf{bool}.\tend$ types a channel as follows: send an
integer, receive a boolean, and close the channel. Due to its simplicity and
expressiveness, the $\pi$-calculus~\cite{milner_calculus_1992,%
sangiorgi_pi-calculus_2003}---the paradigmatic model of concurrency and
interaction---is a widely used setting for studying session types.

In a line of work developed by Caires, Pfenning, Wadler, and several others, the
theory of session types has been given strong logical foundations. Caires and
Pfenning discovered a Curry-Howard correspondence between a form of session
types for the $\pi$-calculus and Girard's \emph{linear
logic}~\cite{girard_linear_1987}: session types correspond to linear logic
propositions, type inference rules to sequent calculus, and communication to cut
reduction~\cite{caires_session_2010}. The resulting session type system ensures
important correctness properties for communicating processes: protocol fidelity,
communication safety, deadlock freedom, and strong normalization.

As in standard logic, there are two ``schools'' of linear logic:
\emph{classical}~\cite{girard_linear_1987} and
\emph{intuitionistic}~\cite{barber_dual_1996}. The differences between classical
and intuitionistic linear logic are known---see,
e.g.,~\cite{chang_judgmental_2003,laurent_around_2018}. This dichotomy also
appears in the logical foundations of session types: while Caires and Pfenning's
correspondence relies on intuitionistic linear logic~\cite{barber_dual_1996}, Wadler
developed a correspondence based on classical linear
logic~\cite{wadler_propositions_2012}. Superficial differences between the
resulting type systems include the number of typing rules (the intuitionistic
system has roughly twice as many rules as the classical system) and the
shape/meaning of typing judgments (in the intuitionistic system, judgments have
a rely-guarantee reading not present in the classical system). In turn, these
differences follow from the way in which each system internalizes duality: the
classical system provides a more explicit account of duality than the
intuitionistic system.

We are interested in going beyond these superficial differences, so as to
establish a   formal comparison between the two type systems. This seems to
us an indispensable step in consolidating the logical foundations of
message-passing concurrency. To our knowledge, the only available comparison is
informal: Caires, Pfenning, and Toninho~\cite{caires_linear_2016} observed that a
more fundamental difference concerns the \emph{locality principle} for shared
channels. The principle states that received channels cannot be used for further
reception, i.e., only the output capability of channels can be
sent~\cite{merro_locality_2000}. In session-based concurrency, shared channels
define services; clients connect to services by sending a linear channel.
Locality of shared channels therefore means that received channels cannot be
used to provide a service. Well-known from a foundational perspective, locality
has been promoted as a sensible principle for distributed implementations of (object-oriented)
languages based on the $\pi$-calculus~\cite{merro_asynchrony_2004}. The
observation in~\cite{caires_linear_2016} is that Caires and Pfenning's
intuitionistic interpretation of session types enforces locality of shared
channels~\cite{caires_linear_2016}, whereas Wadler's classical
interpretation does not: processes that break locality are well-typed
in~\cite{wadler_propositions_2012}.

The existence of a class of processes that is typable in one system but not in
the other immediately frames the desired formal comparison as an expressiveness
question: the type system in~\cite{wadler_propositions_2012} can be considered
to be \emph{more expressive} than the one in~\cite{caires_session_2010}. To
formally examine this question, the first step is defining a basic framework of
reference in which both type systems can be objectively compared. To this end,
we build upon Girard's Logic of Unity~\cite{girard_unity_1993}, which subsumes
classical, intuitionistic and linear logic in one system. In the same spirit, we
develop \emph{United Linear Logic} (\ull): a logic that subsumes classical and
intuitionistic linear logic. Following the Curry-Howard correspondence by Caires
and Pfenning, we interpret \ull as a session type system for the $\pi$-calculus,
dubbed \pull. The class of \pull-typable processes therefore contains processes
induced by type systems derived from both intuitionistic and classical
interpretations of linear logic. Using \pull, we corroborate and formalize
Caires, Pfenning, and Toninho's observation as inclusions between classes of
typable processes: our technical results are that (i)~\pull precisely captures the
class of  processes typable under the classical interpretation and that (ii)~the
class of processes typable under the intuitionistic interpretation is strictly
included in \pull.

This paper is structured as follows. In Section~\ref{sec:ull} we introduce \ull,
explain the Curry-Howard interpretation as the session type system \pull, and
detail the correctness properties for processes derived by typing.
Section~\ref{sec:comparison} formally establishes the differences between the
classical and intuitionistic interpretations of linear logic as session type
systems. Section~\ref{sec:conclusion} concludes the paper.

\section{United Linear Logic as a Session Type System}\label{sec:ull}

In this section, we introduce United Linear Logic (\ull), a logic based on the
linear fragment of Girard's Logic of Unity~\cite{girard_unity_1993}. We present
\ull as a session type system for the
$\pi$-calculus~\cite{milner_calculus_1992,sangiorgi_pi-calculus_2003}, dubbed
\pull, following the Curry-Howard correspondences established by Caires and
Pfenning~\cite{caires_session_2010} and by
Wadler~\cite{wadler_propositions_2012}.

\paragraph{Propositions / Types.}
Propositions in \ull correspond to session types in \pull; they are defined as
follows:

\begin{definition}\label{def:propositions}
    \ull propositions / \pull types are generated by the following grammar:
    \begin{align*}
        A, B ::= \1 \bnf \bot \bnf A \tensor B \bnf A \parr B \bnf A \lolli B
        \bnf \bang A \bnf \whynot A
    \end{align*}
\end{definition}

Session types represent sequences of communication actions that should be
performed along channels. Table~\ref{tbl:propsassessiontypes} gives the
intuitive reading of the interpretation of propositions as session types. Note
that there are two types for reception: \parr from classical and \lolli from
intuitionistic linear logic.

{
    \rowcolors{1}{black!3}{black!5}
    \begin{table}[h]
        \begin{center}
            \begin{tabular}{ll}
                $\1$ and $\bot$
                & Close the channel \\
                $A \tensor B$
                & Send a channel of type $A$ and continue as $B$ \\
                $A \parr B$ and $A \lolli B$
                & Receive a channel of type $A$ and continue as $B$ \\
                $\bang A$
                & Repeatedly provide a service of type $A$ \\
                $\whynot A$
                & Connect to a service of type $A$
            \end{tabular}
    \end{center}
    \caption{Interpretation of \ull propositions as session types}
    \label{tbl:propsassessiontypes}
    \end{table}
}

\ull does not include the additives $A \oplus B$ and $A \mathbin{\&} B$ of
linear logic. Although the Logic of Unity does include these connectives, we
leave them out from \ull (and \pull), because their interpretation as session
types---internal and external choice, respectively---largely coincides in many
presentations of logic-based session types. Therefore they are not particularly
insightful in our formal comparison. They can be easily accommodated in \ull,
with the possibility of choosing between binary choice (as in
e.g.~\cite{caires_session_2010,caires_linear_2016}) and $n$-ary choice (as in
e.g.~\cite{wadler_propositions_2012,caires_towards_2012}).

\paragraph{Duality.}
The duality of \ull propositions is given in Definition~\ref{def:duality}. In
\pull duality is reflected by the intended reciprocity of protocols between two
parties: when a process on one side of a channel sends, the process on the
opposite side must receive, and vice versa.

\begin{definition}\label{def:duality}
    Duality ($\dual{A}$) is given by the following set of equations:
    \begin{align*}
        \dual{\1} &:= \bot
        & \dual{(A \tensor B)} &:= \dual{A} \parr \dual{B}
        & \dual{(\bang A)} &:= \whynot \dual{A} \\
        \dual{\bot} &:= \1
        & \dual{(A \parr B)} &:= \dual{A} \tensor \dual{B}
        & \dual{(\whynot A)} &:= \bang \dual{A}
    \end{align*}
\end{definition}

\noindent It is easy to see that duality is an involution: $\dual{(\dual{A})} =
A$. As usual, we decree that $A \lolli B = \dual{A} \parr B$. From this, we can
derive the relation between $\lolli$ and $\tensor$ by means of their duals:
\begin{align*}
    \dual{(A \lolli B)} &= \dual{(\dual{A} \parr B)} = \dual{(\dual{A})} \tensor
    \dual{B} = A \tensor \dual{B} \\
    \dual{(A \tensor B)} &= \dual{A} \parr \dual{B} = A \lolli \dual{B}
\end{align*}

\paragraph{Processes.}
\pull is a type system for the $\pi$-calculus processes defined as follows:

\begin{definition}\label{def:processterms}
    Process terms are generated by the following grammar:
    \begin{align*}
        P, Q := \0 \bnf \fwd{x}{y} \bnf \nu{x}P \bnf P \| Q \bnf \send{x}{y}.P
        \bnf \recv{x}{y}.P \bnf \serv{x}{y}.P \bnf \send{x}{}.\0 \bnf
        \recv{x}{}.P
    \end{align*}
\end{definition}

\noindent Process constructs for inaction $\0$, channel restriction $\nu{x}P$,
and parallel composition $P \| Q$ have standard readings. The same applies to
constructs for output, input, and replicated input prefixes, which are denoted
$\send{x}{y}.P$, $\recv{x}{y}.P$, and $\serv{x}{y}.P$, respectively. Process
$\fwd{x}{y}$ denotes a forwarder that ``fuses'' channels $x$ and $y$. We
consider also constructs $\send{x}{}.\0$ and $\recv{x}{}.P$, which specify the
explicit closing of channels: their synchronization represents the explicit
de-allocation of linear resources. These constructs result from the non-silent
interpretation of $\1$, which, as explained in~\cite{caires_towards_2012},
leads to a Curry-Howard correspondence that is stronger than correspondences
with silent interpretations of $\1$ (such as those
in~\cite{caires_session_2010,caires_linear_2016}).

\newpage
\paragraph{Structural congruence.}
Processes can be syntactically different, but still exhibit the same behavior.
Such processes are \emph{structurally congruent}, in the sense of the following definition:

\begin{definition}\label{def:structuralcongruence}
    Structural congruence ($\equiv$) is given by the following set of equations,
    where $\equiv_\alpha$ denotes equality up to capture-avoiding
    $\alpha$-conversion, and $\fn(P)$ gives the set of free names in $P$, i.e.\
    the complement of $\bn(P)$: the names in $P$ bound by restriction $\nu{x}$
    and (replicated) input $\recv{x}{y}$ and $\serv{x}{y}$:
    \begin{align*}
        & P \| \0 \equiv P
        && P \| (Q \| R) \equiv (P \| Q) \| R \\
        & P \| Q \equiv Q \| P
        && \fwd{x}{y} \equiv \fwd{y}{x} \\
        & \nu{x}\0 \equiv \0
        && P \equiv_\alpha Q \implies P \equiv Q \\
        & \nu{x}\nu{y}P \equiv \nu{y}\nu{x}P
        && x \notin \fn(P) \implies P \| \nu{x}Q \equiv \nu{x}(P \| Q)
    \end{align*}
\end{definition}

\paragraph{Computation.}
In a Curry-Howard correspondence, computation is related to cut reduction in the
logic. Cut reduction removes cuts from an inference tree, which reduces the size
of the tree without changing the result of the inference. In the correspondence
between linear logic and the $\pi$-calculus, cut reduction is related to
communication, defined by the following reduction relation:

\begin{definition}\label{def:reductionrelation}
    Reduction of process terms ($\red$) is given by the following relation:
    \begin{align*}
        & \send{x}{y}.P \| \recv{x}{z}.Q \red P \| Q\subst{y/z}
        && Q \red Q' \implies P \| Q \red P \| Q' \\
        & \send{x}{y}.P \| \serv{x}{z}.Q \red P \| Q\subst{y/z} \| \serv{x}{z}.Q
        && P \red Q \implies \nu{y}P \red \nu{y}Q \\
        & \send{x}{}.\0 \| \recv{x}{}.Q \red Q
        && P \equiv P' \wedge P' \red Q' \wedge Q' \equiv Q \implies P \red Q \\
        & P \| \fwd{x}{y} \red P\subst{y/x}
    \end{align*}
\end{definition}

\paragraph{Typing inference.}
The inference system of \ull is a sequent calculus with sequents of the form
$\Gamma; \Delta \vdash \Lambda$ in which $\Gamma$ is a collection of
propositions that can  be indefinitely used, and $\Delta$ and $\Lambda$ collect
propositions that must be used linearly (exactly once). With respect to the
Logic of Unity, we added a left rule for $\1$ and a right rule for $\bot$, and
removed rules that allow propositions to switch sides in a sequent.

\pull's typing inference system annotates sequents with process terms and
channel names to form typing judgments of the following form:
\vspace{-7pt}
$$\Gamma; \Delta \vdash P :: \Lambda$$
In this interpretation, $\Gamma$ (resp.\ $\Delta$ and $\Lambda$), the
unrestricted (resp.\ linear) context of $P$, consists of assignments of the form
$x:A$, where $x$ is a channel and $A$ is a proposition/type. In a correct
inference, these contexts together contain exactly the free channel names of the
process term $P$. We write $\cdot$ to denote an empty context. \pull's inference
rules are given in Figure~\ref{fig:ull-inf}. Note that some rules are labeled
with an `$\ast$', which we use to distinguish a class of rules to be used in the
formal comparison in the next section.

\begin{figure}
    \input{ull.tex}
    \caption{The $\ull$ inference rules / type system}\label{fig:ull-inf}
\end{figure}

\paragraph{Cut reduction and identity expansion.}
Caires, Pfenning, and Toninho~\cite{caires_towards_2012} showed that the
validity of session type interpretations of linear logic propositions can be
demonstrated by checking that cut reductions in typing inferences do correspond
to reductions of processes, as well as by showing that the identity axiom of any
type can be expanded to a larger process term with forwarding of a smaller type.
Following this approach, \pull can be shown valid for all reductions, using
$\cut \rgt$ as well as $\cut \lft$, and expansions, using $\id \rgt$ as well as
$\id \lft$.

\paragraph{Correctness properties.}
\renewcommand{\proofname}{Proof (sketch)}
As is usual for Curry-Howard correspondences for concurrency, the cut
elimination property of linear logic means that \pull has the
\emph{soundness/subject reduction} (Theorem~\ref{thm:subjectreduction}) and
\emph{progress} (Theorem~\ref{thm:progress}) properties.

\newpage
\begin{theorem}\label{thm:subjectreduction}
    If $\Gamma; \Delta \vdash P :: \Lambda$ and $P \red Q$, then $\Gamma; \Delta
    \vdash Q :: \Lambda$.
\end{theorem}

\begin{proof}
    By induction on the structure of the proof of $P$, using cut reduction.
\end{proof}

\begin{definition}
    For any process $P$, $\live(P)$ if and only if there are processes $Q$ and
    $R$ and channels $\tilde{n}$ such that $P \equiv \nu{\tilde{n}}(\pi.Q \|
    R)$, where $\pi \in \{\send{x}{y}, \recv{x}{y}, \send{x}{}, \recv{x}{}\}$.
\end{definition}

\begin{theorem}\label{thm:progress}
    If $\cdot; \cdot \vdash P :: \cdot$ and $\live(P)$, then there exists $Q$
    such that $P \red Q$.
\end{theorem}

\begin{proof}
    The liveness assumption tells us that $P$ is a parallel composition of a
    process $Q$ that is guarded by some non-persistent communication $\pi$ and
    some other process $R$. The fact that $P$'s proof has an empty context
    allows us to infer that $\pi.Q$ and $R$ must have been composed using a
    $\cut$ rule and that $R$ must be guarded by a prefix that is dual to $\pi$.
    Therefore, a reduction can take place.
\end{proof}
\renewcommand{\proofname}{Proof}

\section{Comparing Expressivity through United Linear Logic}%
\label{sec:comparison}

\pull is a suitable framework for a rigorous comparison of the expressivity of
session type systems based on linear logic. In this section, we compare the
class of processes typable in \pull to the classes of processes typable in a
classical and intuitionistic interpretation of linear logic. For this comparison
to be fair, the differences between these classes need to come only from typing,
so the process languages need to be the same. This means that our classical and
intuitionistic interpretations need to type process terms as given in the
previous section, with the same features as those in \pull: explicit closing, a
separate unrestricted context, and identity as forwarding.

\begin{table}[]
\newcolumntype{Y}{>{\centering\arraybackslash}X}
\centering
\begin{tabularx}{.85\textwidth}{lYYY}
 & \textbf{Explicit closing}
 & \textbf{Separate \mbox{unrestricted} context}
 & \textbf{Identity as \mbox{forwarding}} \\ \hline
\pcll~\cite{caires_linear_2016}   & No  & Yes & No  \\
\cp~\cite{wadler_propositions_2012}     & Yes & No  & Yes \\
\pcllcp & Yes & Yes & Yes
\end{tabularx}
\caption{Feature comparison of three session type interpretations of classical
linear logic}\label{tbl:pcll_cp_pcllcp}
\end{table}

\begin{figure}
    \newcommand{\sepspace}{53pt}

    \vspace{\sepspace}
    \input{pill.tex}
    \caption{The \pill type system}\label{fig:ill-inf}%

    \vspace{\sepspace}
    \input{pcllcp.tex}
    \caption{The \pcllcp type system}\label{fig:cll-inf}
\end{figure}

Our intuitionistic and classical session type interpretations of linear logic
are denoted \pill and \pcllcp, respectively. Their respective rules are given in
Figure~\ref{fig:ill-inf} and Figure~\ref{fig:cll-inf}. \pill is based on the
work by Caires, Pfenning and Toninho~\cite{caires_towards_2012}.

It is worth noticing that, because of the features we require for our type
systems, we could not directly base our classical interpretation on \pcll by
Caires, Pfenning and Toninho~\cite{caires_linear_2016} nor on Wadler's
\cp~\cite{wadler_propositions_2012}. Therefore, we have designed \pcllcp as a
combination of features from \pcll and \cp. Table~\ref{tbl:pcll_cp_pcllcp}
compares these features in \pcll, \cp and \pcllcp; the differences are merely
superficial:
\begin{itemize}
    \item Explicit closing of sessions concerns a non-silent interpretation of
        $\1$ and $\bot$ in the logic that entails a  reduction on processes
        (which, in turn, corresponds to cut reduction). In contrast, implicit
        closing is due to a silent interpretation and corresponds to
        (structural) congruences in processes.
    \item Sequents with a separate unrestricted context are of the form $P
        \vdash \Gamma; \Delta$, which can also be written as $P \vdash \Delta,
        \Gamma'$ where $\Gamma'$ contains only types of the form $!A$.
    \item The identity axiom can be interpreted as the forwarding process, which
        enables to account for polymorphism. The forwarding process, however, is
        not usually present in session $\pi$-calculi.
\end{itemize}

\paragraph{Judgments.}
Before we study the expressivity of these three systems, it is important to take
note of the difference in the forms of their typing judgments. For \pull,
\pcllcp, and \pill, respectively, they are as follows:

\vspace{-7pt}
\begin{minipage}[T]{.31\textwidth}
    $$\Gamma; \Delta \vdash P :: \Lambda$$
\end{minipage}%
\hfill%
\begin{minipage}[T]{.31\textwidth}
    $$P \vdcll \Gamma; \Delta$$
\end{minipage}%
\hfill%
\begin{minipage}[T]{.31\textwidth}
    $$\Gamma; \Delta \vdill x:A$$
\end{minipage}
\smallskip

\noindent \pcllcp has one-sided sequents, whereas \pill has two-sided sequents
(like \pull), but with exactly one channel/type pair on the right of the
turnstile. Compare, for example, \pcllcp's inference rule for \tensor and the
$\tensor \rgt$ rules for \pill and \pull (see Figure~\ref{fig:ull-inf}):

\hspace{-25pt}%
\begin{minipage}[T]{.5\textwidth}
    \begin{prooftree}
        \infAss
            {$P \vdcll \Gamma; \Delta, y:A$}
        \infAss
            {$Q \vdcll \Gamma; \Delta', x:B$}
        \infBin
            {$\nu{y}\send{x}{y}.(P \| Q) \vdcll \Gamma; \Delta, \Delta', x:A
            \tensor B$}
            {$\tensor$}
    \end{prooftree}
\end{minipage}%
\hspace{0pt}%
\begin{minipage}[T]{.5\textwidth}
    \begin{prooftree}
        \infAss
            {$\Gamma; \Delta \vdill P :: y:A$}
        \infAss
            {$\Gamma; \Delta' \vdill Q :: x:B$}
        \infBin
            {$\Gamma; \Delta, \Delta' \vdill \nu{y}\send{x}{y}.(P \| Q) ::
            x:A \tensor B$}
            {$\tensor \rgt$}
    \end{prooftree}
\end{minipage}
\smallskip

\paragraph{Locality.}
Locality is a well-known principle in concurrency research~\cite{merro_locality_2000}. The idea is that
freshly created channels are \emph{local}. Local channels are \emph{modifiable},
in the sense that they can be used for inputs. Once a channel has been
transmitted to another location, it becomes \emph{non-local}, and thus
\emph{immutable}: it can only be used for outputs---inputs are no longer allowed. This makes locality particularly
relevant for giving formal semantics to distributed programming languages; a prime example is
the \emph{join calculus}~\cite{fournet_calculus_1996}, whose theory relies on (and is deeply influenced by) the locality principle~\cite{fournet_bisimulations_2001}.

\pill guarantees locality for shared channels: a server can be defined using a
replicated input, so the channel on which this server would be provided cannot
be received earlier in the process. Consider the following example, taken
from~\cite{caires_linear_2016}:
$$\nu{x}(\recv{x}{y}.\serv{y}{z}.P \| \nu{q}\send{x}{q}.Q)$$
Let us attempt to find a typing for $P$ in this process using \pill. We can
apply the cut rule to split the parallel composition, of which the left
component is $\recv{x}{y}.\serv{y}{z}.P$. Now, there are two rules we can apply
(read bottom-up):

\hspace{-13pt}%
\begin{minipage}{.5\textwidth}
    \begin{prooftree}
        \infAss
            {$\Gamma; \Delta, y:A, x:B \vdill \serv{y}{z}.P :: w:C$}
        \infUn
            {$\Gamma; \Delta, x:A \tensor B \vdill \recv{x}{y}.\serv{y}{z}.P ::
            w:C$}
            {$\tensor \lft$}
    \end{prooftree}
\end{minipage}%
\hspace{-6pt}%
\begin{minipage}{.5\textwidth}
    \begin{prooftree}
        \infAss
            {$\Gamma; \Delta, y:A \vdill \serv{y}{z}.P :: x:B$}
        \infUn
            {$\Gamma; \Delta \vdill \recv{x}{y}.\serv{y}{z}.P :: x:A \lolli B$}
            {$\lolli \rgt$}
    \end{prooftree}
\end{minipage}
\smallskip

\noindent In both cases, the received channel $y$ ends up on the left of the
turnstile. There are no rules in \pill to define a service on a channel on the
left and there is no way to move the channel to the right. Hence, we cannot find
a typing for $P$ in \pill. In contrast, this process can be typed in \pcllcp as
follows, given $P=P'\subst{x/u}$:

\vspace{-25pt}
\begin{prooftree}
    \infAss
        {$P' \vdcll \Gamma, u:B; z:A$}
    \infUn
        {$\serv{y}{z}.P' \vdcll \Gamma, u:B; y:\bang A$}
        {$\bang$}
    \infUn
        {$\serv{y}{z}.P \vdcll \Gamma; y:\bang A, x:\whynot B$}
        {$\whynot$}
    \infUn
        {$\recv{x}{y}.\serv{y}{z}.P \vdcll \Gamma; x:(\bang A) \parr (\whynot
        B)$}
        {$\parr$}
\end{prooftree}

\paragraph{Symmetry.}
The rely-guarantee distinction of \pill is what makes it enforce locality for
shared channels. The distinction is visible in \pill's distinction between left
and right in its judgments. Despite the two-sidedness of its sequents, \pull can
also type non-local processes. This is due to the full symmetry in the inference
rules: anything that can be done on the left of the turnstile can be done on the
right. As we will show formally in the rest of this section, this symmetry
corresponds to the single-sidedness of \pcllcp: \pull and \pcllcp can type
exactly the same processes. We will also show formally that \pull, and thus
\pcllcp, can type more processes than \pill, because of the restriction of the
right side of the turnstile to exactly one channel/type pair, making \pill an
asymmetrical typing system.

\newpage
\paragraph{Formal results.}
\renewcommand{\proofname}{Proof (sketch)}
Our formal results rely on the sets of processes typable in the three typing
systems, given in Definition~\ref{def:typablesets}.

\begin{definition}\label{def:typablesets}
    Let $\mathbb{P}$ denote the set of all processes induced by
    Definition~\ref{def:processterms}. Then
    \begin{align*}
        \U &= \{P \in \mathbb{P} \mid \exists \Gamma, \Delta, \Lambda \text{
        such that } \Gamma; \Delta \vdash P :: \Lambda\}, \\
        \C &= \{P \in \mathbb{P} \mid \exists \Gamma, \Delta \text{ such that }
        P \vdcll \Gamma; \Delta\}, \\
        \I &= \{P \in \mathbb{P} \mid \exists \Gamma, \Delta, x \in \fn(P), A
        \text{ such that } \Gamma; \Delta \vdill P :: x:A\}.
    \end{align*}
\end{definition}

Our first observation is that all \pull-typable processes are \pcllcp-typable.
Moreover, the reverse is true as well. For the latter fact we need to convert a
typing inference with one-sided sequents to a two-sided system, which means that
we need the ability to control the side of the turnstile specific propositions
end up on. This is taken care of by Lemma~\ref{lem:ulltotheleft}, which is an
extension of~\cite[Prop.5.1,~p.19]{caires_linear_2016} to include exponential
types (\bang and \whynot). The main result, Theorem~\ref{thm:cllandull}, is that
\pcllcp and \pull type exactly the same processes.

\begin{lemma}\label{lem:ulltotheleft}
    For any typing contexts $\Gamma, \Delta, \Lambda, \Pi$ and process $P$ such
    that $\Gamma; \Delta \vdash P :: \Lambda, \Pi$, we have $\Gamma; \Delta,
    \dual{\Pi} \vdash P :: \Lambda$.
\end{lemma}

\begin{proof}
    By induction on the structure of the type inference tree. If the last
    inferred proposition is to be moved to the left, after type inversion, the
    appropriate left rule can be used where a right rule was used. Other
    propositions can be moved using the induction hypothesis.
\end{proof}

\begin{theorem}\label{thm:cllandull}
    $\U = \C$.
\end{theorem}

\begin{proof}
    There are two things to prove:
    \begin{enumerate*}[label=(\roman*)]
        \item $\U \subseteq \C$, and
        \item $\C \subseteq \U$.
    \end{enumerate*}
    (i) can be proven by induction on the structure of the typing inferences.
    The idea is that for every rule of \pull there is an analogous rule in
    \pcllcp. As for (ii), for any $P$ such that $P \vdcll \Gamma; \Delta$, it
    suffices to show that there are $\Delta_L, \Delta_R = \Delta$ such that
    $\dual{\Gamma}, \dual{(\Delta_L)} \vdash P :: \Delta_R$. This can be done by
    induction on the structure of the typing inference of $P$. After type
    inversion, the induction hypothesis can be used, which moves channel/type
    pairs to either side of the turnstile. This results in many subcases, each
    of which can be solved with appropriate applications of \pull rule analogous
    to the last used \pcllcp rule, using Lemma~\ref{lem:ulltotheleft} in some
    cases.
\end{proof}

The comparison between \pull and \pill can be done similarly. However, it is
more interesting to examine the set of inference rules. If we restrict all
sequents in a \pull typing inference to have exactly one channel/type on the
right of the turnstile, we end up with a subset of usable inference rules: those
marked with an `$\ast$' in Figure~\ref{fig:ull-inf}. Upon further examination,
we see that this set of rules coincides exactly with the set of inference rules
for \pill. The consequence is that \pull can type all \pill-typable processes.
Finally, any \pull/\pcllcp-typable process violating the locality principle
suffices to show the second main result, Theorem~\ref{thm:illandull}: \pill is
less expressive than \pull. An important corollary of
Theorem~\ref{thm:cllandull} and Theorem~\ref{thm:illandull} is that \pill is
less expressive than \pcllcp, confirming the observation by Caires, Pfenning, and
Toninho~\cite{caires_linear_2016}.

\begin{definition}[\rdegree]
    The size of the right side of a \pull sequent is the sequent's
    \emph{\rdegree}. Given contexts $\Gamma, \Delta, \Lambda$ and process $P$,
    the \ull sequent $\Gamma; \Delta \vdash P :: \Lambda$ has \rdegree
    $|\Lambda|$.
\end{definition}

\newpage
\begin{theorem}\label{thm:illandull}
    $\I \subset \U$.
\end{theorem}

\begin{proof}
    Let
    $$\U^\ast := \{P \in \U \mid \exists \Gamma, \Delta, A \text{ s.t.\ }
        \Gamma; \Delta \vdash P :: x:A \text{ with a proof tree containing only
    sequents of \rdegree 1}\}.$$
    By induction on the structure of the typing inference, it can be shown that
    all processes in $\U^\ast$ have typing inferences using only those rules in
    Figure~\ref{fig:ull-inf} marked with an $\ast$. It follows by contradiction
    that the last applied rule cannot be without an $\ast$, and the usage of
    $\ast$-marked rules in the rest of the inference follows from the induction
    hypothesis. The $\ast$-marked rules with \rdegree 1 coincide exactly with
    the typing inference rules of \pill, and so it follows that $\I = \U^\ast$.
    Clearly, $\U^\ast$ is a subset of $\U$. Now, given $u:B; \cdot \vdash P ::
    z:A$, there are several ways to type the process $\recv{x}{y}.\serv{y}{z}.P$
    in \pull, but all of them use sequents of \rdegree different from 1, so it
    is not in $\U^\ast$. Hence, $\U^\ast \subset \U$, confirming the thesis.
\end{proof}
\renewcommand{\proofname}{Proof}

\section{Conclusion}\label{sec:conclusion}

Using Girard's Logic of Unity~\cite{girard_unity_1993} as a basis, we have
developed United Linear Logic as a means to formally compare the session type
systems derived from concurrent interpretations of classical and intuitionistic
linear logic. Much as Logic of Unity is a logic that encompasses classical,
intuitionistic and linear logic, the session type system interpretation of \ull
(dubbed \pull), can type all \pcllcp- and \pill-typable processes. This allows
us to compare type systems based on different linear logics.

In \pill, judgments distinguish between several channels whose behavior is being
relied on (on the left of the judgment) and a single behavior provided along a
designated channel (on the right). To retain this rely-guarantee reading, \pull
uses two-sided sequents in its typing inferences. However, \pull does not
distinguish between the sides of its sequents; it is fully symmetrical. This
symmetry allows it to mimic the single-sidedness of \pcllcp's sequents: placing
a mirror besides \pcllcp's inference rules reveals the inference rules of \pull.
Similarly, \pill can be recovered from \pull by restricting the right side of
its sequents to exactly one channel. Not all inference rules remain usable,
resulting in an asymmetrical system that coincides exactly with \pill.

Our results formally corroborate the observation by Caires, Pfenning and
Toninho: the difference between session type systems based on classical and
intuitionistic linear logic is in the enforcement of
locality~\cite{caires_linear_2016}. \pcllcp is able to type non-local processes,
because it does not distinguish between linear channels, received or not,
because of its single-sided sequents. Similarly, \pull can type non-local
processes because of its symmetry. Restricting \pull into \pill removes exactly
those rules needed to violate the locality principle. This reveals that \pill
respects locality because of its asymmetry.


\paragraph*{Acknowledgements.}
We are grateful to Joseph Paulus for initial discussions.
We would also like to thank the anonymous reviewers for their suggestions, which were helpful to improve the presentation.

\label{sec:bib}
\bibliographystyle{eptcs}
\bibliography{refs}

%
%
%
%
%
%

\end{document}